\documentclass[letterpaper, 10 pt, conference]{ieeeconf}  

\IEEEoverridecommandlockouts                              
                                                         
\usepackage{xcolor}
\usepackage{cite}
\usepackage{amsmath,amssymb,amsfonts}
\usepackage{algorithmic}
\usepackage{algorithm}
\usepackage{graphicx}
\usepackage{subfig}
\usepackage{textcomp}

\graphicspath{{pics/}}

\usepackage{mathtools}

\usepackage[acronym, style=alttree, toc=true, shortcuts, xindy, nomain, nonumberlist]{glossaries}

\usepackage{cleveref}

\usepackage{booktabs}

\usepackage{hyphenat}
\usepackage{bm}

\usepackage{siunitx}  

\newcommand{\bs}[1]{\boldsymbol{#1}}
\newcommand{\ol}[1]{\overline{#1}}

\newcommand{\norm}[1]{\left\lVert#1\right\rVert} 
\newcommand{\setdef}[2][]{
	\left\{
		\ifblank{#1}{}{#1 \hspace{.1cm} \middle| \hspace{.1cm}}
		#2
	\right\}
} 

\newcommand{\lr}[1]{\left(#1\right)} 

\DeclareMathOperator\erf{erf} 

\DeclarePairedDelimiterXPP\onenorm[1]{}\lVert\rVert{_1}{\ifblank{#1}{\:\cdot\:}{#1}} 
\DeclarePairedDelimiterXPP\twonorm[1]{}\lVert\rVert{_2}{\ifblank{#1}{\:\cdot\:}{#1}} 

\newtheorem{theorem}{Theorem}
\newtheorem{assumption}{Assumption}

\newtheorem{definition}{Definition}
\newtheorem{lemma}{Lemma}

\newtheorem{objective}{Objective}

\newglossary{symbols}{sym}{sbl}{List of Symbols}
\newglossary{symbols2}{sym2}{sbl2}{List of Symbols (no glossary)}
\newglossary{notation}{not}{nt}{Notation}

\glsaddkey{dimension}{\glsentrytext{\glslabel}}{\glsentryd}{\GLsentryd}{\glsd}{\Glsd}{\GLSd}
\glsaddkey{body}{\glsentrytext{\glslabel}}{\glsentrybody}{\GLsentrybody}{\glsbody}{\Glsbody}{\GLSbody}

\newacronym{MPC}{MPC}{Model Predictive Control}
\newacronym{MPCa}{MPC algorithm}{MPC algorithm}
\newacronym{RMPC}{RMPC}{Robust Model Predictive Control}
\newacronym{SMPC}{SMPC}{Stochastic Model Predictive Control}
\newacronym{SCMPC}{SCMPC}{Scenario Model Predictive Control}
\newacronym{MILP}{MILP}{Mixed Integer Linear Program}
\newacronym{PIT}{PIT}{Pointwise\hyp{}In\hyp{}Time}
\newacronym{POMDP}{POMDP}{Partially Observable Markov Decision Process}
\newacronym{MDP}{MDP}{Markov Decision Process}
\newacronym{KKT}{KKT}{Karush\hyp{}Kuhn\hyp{}Tucker}
\newacronym{ev}{EV}{ego vehicle}
\newacronym{tv}{TV}{target vehicle}
\newacronym{cog}{CoG}{center of gravity}
\newacronym{ol}{OL}{Open-Loop}
\newacronym{cl}{CL}{Closed-Loop}
\newacronym{ocp}{OCP}{optimal control problem}
\newacronym{pog}{PG}{\textit{Probabilistic Grid}}
\newacronym{bog}{BG}{\textit{Binary Grid}}
\newacronym{lk}{LK}{lane keeping}
\newacronym{lc}{LC}{lane changing}
\newacronym{og}{OG}{Occupancy Grid}
\newacronym{ISS}{ISS}{input-to-state stable}

\newglossaryentry{x}{type=symbols2,
	sort={x},
	dimension={\ensuremath{\mathbb{R}^{\gls{nx}}}},
	name={\ensuremath{\bm{x}}},
	description={State}
}
\newglossaryentry{xt}{type=symbols,
	sort={xt},
	name={\ensuremath{\bm{x}_t}},
	description={State at time step $t$}
}
\newglossaryentry{x0}{type=symbols,
	sort={x0},
	name={\ensuremath{\bm{x}_0}},
	description={Initial state}
}
\newglossaryentry{x1}{type=symbols2,
	sort={x1},
	name={\ensuremath{\bm{x}_1}},
	description={State at time\prediction step $1$}
}
\newglossaryentry{xk}{type=symbols,
	sort={xk},
	name={\ensuremath{\bm{x}_k}},
	description={State at prediction step $k$}
}
\newglossaryentry{xk1}{type=symbols2,
	sort={xk1},
	name={\ensuremath{\bm{x}_{k+1}}},
	description={State at prediction step $k+1$}
}
\newglossaryentry{xt1}{type=symbols2,
	sort={xt1},
	name={\ensuremath{\bm{x}_{t+1}}},
	description={State at time step $t+1$}
}
\newglossaryentry{xtk}{type=symbols2,
	sort={xt1k},
	name={\ensuremath{\bm{x}_{t+k|t}}},
	description={Predicted state for time step $t+k$, at time step $t$}
}

\newglossaryentry{u}{type=symbols,
	sort={u},
	dimension={\ensuremath{\mathbb{R}^{\gls{nu}}}},
	name={\ensuremath{\bm{u}}},
	description={Input}
}
\newglossaryentry{us}{type=symbols,
	sort={u},
	name={\ensuremath{\bm{u}^{\text{s}}}},
	description={Input}
}
\newglossaryentry{ub}{type=symbols,
	sort={u},
	name={\ensuremath{\bm{u}^{\text{b}}}},
	description={Input}
}
\newglossaryentry{ut}{type=symbols,
	sort={ut},
	name={\ensuremath{\bm{u}_t}},
	description={Input at time step $t$}
}
\newglossaryentry{ust}{type=symbols,
	sort={ut},
	name={\ensuremath{\bm{u}^{\text{s}}_t}},
	description={Input at time step $t$}
}
\newglossaryentry{ubt}{type=symbols,
	sort={ut},
	name={\ensuremath{\bm{u}^{\text{b}}_t}},
	description={Input at time step $t$}
}
\newglossaryentry{u0}{type=symbols,
	sort={u0},
	name={\ensuremath{\bm{u}_0}},
	description={Initial input}
}
\newglossaryentry{u1}{type=symbols2,
	sort={u1},
	name={\ensuremath{\bm{u}_1}},
	description={Input at time/prediction step $1$}
}
\newglossaryentry{uk}{type=symbols,
	sort={uk},
	name={\ensuremath{\bm{u}_k}},
	description={Input at prediction step $k$}
}
\newglossaryentry{uk1}{type=symbols2,
	sort={uk1},
	name={\ensuremath{\bm{u}_{k+1}}},
	description={Input at prediction step $k+1$}
}
\newglossaryentry{ut1}{type=symbols2,
	sort={ut1},
	name={\ensuremath{\bm{u}_{t+1}}},
	description={Input at time step $t+1$}
}
\newglossaryentry{utk}{type=symbols2,
	sort={utk},
	name={\ensuremath{\bm{u}_{t+k|t}}},
	description={Predicted input at time step $t+k$}
}
\newglossaryentry{ustk}{type=symbols2,
	sort={utk},
	name={\ensuremath{\bm{u}^{\text{s}}_{t+k|t}}},
	description={Predicted input at time step $t+k$}
}
\newglossaryentry{ubtk}{type=symbols2,
	sort={utk},
	name={\ensuremath{\bm{u}^{\text{b}}_{t+k|t}}},
	description={Predicted input at time step $t+k$}
}
\newglossaryentry{Ut}{type=symbols,
	sort={Ut},
	name={\ensuremath{\bm{U}_t}},
	description={Input at time step $t$}
}
\newglossaryentry{Ust}{type=symbols,
	sort={Ut},
	name={\ensuremath{\bm{U}^{\text{s}}_t}},
	description={Input at time step $t$}
}
\newglossaryentry{Ubt}{type=symbols,
	sort={Ut},
	name={\ensuremath{\bm{U}^{\text{b}}_t}},
	description={Input at time step $t$}
}

\newglossaryentry{ut2}{type=symbols,
	sort={zz},
	name={\ensuremath{\bs{\nu}_t}},
	description={New input variable at time step $t$}
}

\newglossaryentry{w}{type=symbols,
	sort={w},
	dimension={\ensuremath{\mathbb{R}^{\gls{nw}}}},
	name={\ensuremath{\bm{w}}},
	description={Uncertainty}
}
\newglossaryentry{wt}{type=symbols,
	sort={wt},
	name={\ensuremath{\bm{w}_t}},
	description={Uncertainty at time step $t$}
}
\newglossaryentry{w0}{type=symbols,
	sort={w0},
	name={\ensuremath{\bm{w}_0}},
	description={Initial uncertainty}
}
\newglossaryentry{w1}{type=symbols2,
	sort={w1},
	name={\ensuremath{\bm{w}_1}},
	description={Uncertainty at time/prediction step $1$}
}
\newglossaryentry{wk}{type=symbols,
	sort={wk},
	name={\ensuremath{\bm{w}_k}},
	description={Uncertainty at prediction step $k$}
}
\newglossaryentry{wk1}{type=symbols2,
	sort={wk1},
	name={\ensuremath{\bm{w}_{k+1}}},
	description={Uncertainty at prediction step $k+1$}
}
\newglossaryentry{wt1}{type=symbols2,
	sort={wt1},
	name={\ensuremath{\bm{w}_{t+1}}},
	description={Uncertainty at time step $t+1$}
}
\newglossaryentry{wtk}{type=symbols2,
	sort={wtk},
	name={\ensuremath{\bm{w}_{t+k|t}}},
	description={Predicted uncertainty at time step $t+k$}
}

\newglossaryentry{w2}{type=symbols2,
	sort={w},
	dimension={\ensuremath{\mathbb{R}^{\gls{nw2}}}},
	name={\ensuremath{\tilde{\bm{w}}}},
	description={Uncertainty}
}
\newglossaryentry{w2t}{type=symbols,
	sort={wt},
	name={\ensuremath{\tilde{\bm{w}}_t}},
	description={Uncertainty at time step $t$}
}
\newglossaryentry{w20}{type=symbols,
	sort={w0},
	name={\ensuremath{\tilde{\bm{w}}_0}},
	description={Initial uncertainty}
}
\newglossaryentry{w21}{type=symbols2,
	sort={w1},
	name={\ensuremath{\tilde{\bm{w}}_1}},
	description={Uncertainty at time/prediction step $1$}
}
\newglossaryentry{w2k}{type=symbols,
	sort={wk},
	name={\ensuremath{\tilde{\bm{w}}_k}},
	description={Uncertainty at prediction step $k$}
}
\newglossaryentry{w2k1}{type=symbols2,
	sort={wk1},
	name={\ensuremath{\tilde{\bm{w}}_{k+1}}},
	description={Uncertainty at prediction step $k+1$}
}
\newglossaryentry{w2t1}{type=symbols2,
	sort={wt1},
	name={\ensuremath{\tilde{\bm{w}}_{t+1}}},
	description={Uncertainty at time step $t+1$}
}
\newglossaryentry{w2tk}{type=symbols2,
	sort={wtk},
	name={\ensuremath{\tilde{\bm{w}}_{t+k|t}}},
	description={Predicted uncertainty at time step $t+k$}
}

\newglossaryentry{nx}{type=symbols,
	sort={nx},
	name={\ensuremath{n_{\bm{x}}}},
	description={State dimension}
}
\newglossaryentry{nu}{type=symbols,
	sort={nx},
	name={\ensuremath{n_{\bm{u}}}},
	description={Input dimension}
}
\newglossaryentry{nw}{type=symbols,
	sort={nw},
	name={\ensuremath{n_{\bm{w}}}},
	description={Uncertainty dimension}
}
\newglossaryentry{nw2}{type=symbols2,
	sort={nw},
	name={\ensuremath{n_{\tilde{\bm{w}}}}},
	description={Uncertainty dimension}
}

\newglossaryentry{Wlim}{type=symbols,
	sort={W},
	name={\ensuremath{\mathcal{W}}},
	description={Uncertainty bound}
}
\newglossaryentry{Ulim}{type=symbols,
	sort={U},
	name={\ensuremath{\mathcal{U}}},
	description={Input bound}
}
\newglossaryentry{Xlim}{type=symbols,
	sort={X},
	name={\ensuremath{\mathcal{X}}},
	description={State bound}
}
\newglossaryentry{Sigmaw}{type=symbols,
	sort={zz18w},
	name={\ensuremath{\bs{\Sigma}_{\gls{w}}}},
	description={Covariance matrix of uncertainty}
}
\newglossaryentry{Sigmaet}{type=symbols,
	sort={zz18w},
	name={\ensuremath{\bs{\Sigma}^{\bm{e}}_{t}}},
	description={Covariance matrix of uncertainty}
}
\newglossaryentry{Sigmaek}{type=symbols,
	sort={zz18w},
	name={\ensuremath{\bs{\Sigma}^{\bm{e}}_{k}}},
	description={Covariance matrix of uncertainty}
}
\newglossaryentry{Sigmaet1}{type=symbols,
	sort={zz18w},
	name={\ensuremath{\bs{\Sigma}^{\bm{e}}_{t+1}}},
	description={Covariance matrix of uncertainty}
}
\newglossaryentry{Sigmaek1}{type=symbols,
	sort={zz18w},
	name={\ensuremath{\bs{\Sigma}^{\bm{e}}_{k+1}}},
	description={Covariance matrix of uncertainty}
}
\newglossaryentry{Sigmae0}{type=symbols,
	sort={zz18w},
	name={\ensuremath{\bs{\Sigma}^{\bm{e}}_{0}}},
	description={Covariance matrix of uncertainty}
}

\newglossaryentry{sysA}{type=symbols,
	sort={A},
	name={\ensuremath{\bm{A}}},
	description={System matrix}
}
\newglossaryentry{sysB}{type=symbols,
	sort={B},
	name={\ensuremath{\bm{B}}},
	description={Input matrix}
}
\newglossaryentry{sysG}{type=symbols,
	sort={G},
	name={\ensuremath{\bm{G}}},
	description={Uncertainty matrix}
}

\newglossaryentry{Vf}{type=symbols,
	sort={Vf},
	name={\ensuremath{V_{\text{f}}}},
	description={Terminal cost}
}
\newglossaryentry{Nb}{type=symbols,
	sort={Vf},
	name={\ensuremath{N^{\text{b}}}},
	description={Backup MPC horizon}
}
\newglossaryentry{rp}{type=symbols,
	sort={rp},
	name={\ensuremath{\beta}},
	description={Risk parameter}
}
\newglossaryentry{tightening}{type=symbols2,
	sort={zz3},
	name={\ensuremath{\gamma_{\text{cc},k}}},
	description={Constraint tightening}
}
\newglossaryentry{fl_backup}{type=symbols2,
	sort={?},
	name={\ensuremath{\gls{ub}(\gls{xt})}},
	description={MPC feedback law}
}
\newglossaryentry{fl_smpc}{type=symbols2,
	sort={?},
	name={\ensuremath{\gls{us}(\gls{xt})}},
	description={SMPC feedback law}
}
\newglossaryentry{fl_mpc}{type=symbols2,
	sort={?},
	name={\ensuremath{\gls{u}(\gls{xt})}},
	description={MPC feedback law}
}
\newglossaryentry{roa}{type=symbols2,
	sort={X0},
	name={\ensuremath{\mathcal{X}_0}},
	description={Region of attraction}
}
\newglossaryentry{Xf}{type=symbols2,
	sort={X0},
	name={\ensuremath{\mathcal{X}_{\text{f}}}},
	description={Region of attraction}
}
\newglossaryentry{ocp_smpc}{type=symbols,
	sort={Ps},
	name={\ensuremath{\mathbb{P}^{\text{s}}(\gls{xt})}},
	description={SMPC optimal control problem}
}
\newglossaryentry{ocp_backup}{type=symbols,
	sort={Ps},
	name={\ensuremath{\mathbb{P}^{\text{b}}(\gls{xt})}},
	description={Backup MPC optimal control problem}
}
\newglossaryentry{ocp_backup1}{type=symbols2,
	sort={Ps},
	name={\ensuremath{\mathbb{P}^{\text{b}}(\gls{xt1})}},
	description={Backup MPC optimal control problem}
}
\newglossaryentry{ocp_backup_constraint}{type=symbols2,
	sort={Ps},
	name={\ensuremath{\tilde{\mathbb{P}}^{\text{b}}(\gls{xt})}},
	description={Backup MPC optimal control problem with first step constraint}
}

\newglossaryentry{id}{type=symbols,
	sort={I},
	name={\ensuremath{\mathbb{I}}},
	description={Identity matrix}
}

\title{\LARGE \bf
Safe Stochastic Model Predictive Control
}

\author{T. Br\"udigam, R. Jacumet, D. Wollherr, and M. Leibold
\thanks{The authors are with the Chair of Automatic Control Engineering at the Technical University of Munich, Munich, Germany.
{\tt\small \{tim.bruedigam; r.jacumet; dw; marion.leibold\}@tum.de}} 
}

\begin{document}

\maketitle
\thispagestyle{empty}
\pagestyle{empty}

\begin{abstract}
Combining efficient and safe control for safety-critical systems is challenging. Robust methods may be overly conservative, whereas probabilistic controllers require a trade-off between efficiency and safety. In this work, we propose a safety algorithm that is compatible with any stochastic Model Predictive Control method for linear systems with additive uncertainty and polytopic constraints. This safety algorithm allows to use the optimistic control inputs of stochastic Model Predictive Control as long as a safe backup planner can ensure safety with respect to satisfying hard constraints subject to bounded uncertainty. Besides ensuring safe behavior, the proposed stochastic Model Predictive Control algorithm guarantees recursive feasibility and input-to-state stability of the system origin. The benefits of the safe stochastic Model Predictive Control algorithm are demonstrated in a numerical simulation, highlighting the advantages compared to purely robust or stochastic predictive controllers.
\end{abstract}

\section{Introduction}
\label{sec:introduction}

\vspace{-11cm}
\mbox{\small 
This~work~has~been~accepted~to~the~IEEE~2022~Conference~on~Decision~and~Control.}
\vspace{10.2cm}

Designing controllers for safety-critical systems requires considering two major challenges. Safety must be ensured for a system subject to uncertainty, and the controller should reduce conservatism to enable efficient system behavior, i.e., maximizing desired objectives. As it is possible to define safety via input and state constraints, \gls{MPC} is a suitable method to control safety-critical systems subject to uncertainty. 

\gls{RMPC} handles system uncertainty in a robust, but conservative way \cite{BemporadMorari1999}, where tube-based MPC is the most common approach 
\cite{LangsonEtalMayne2004, KoehlerEtalAllgoewer2021}. 
Stability and recursive feasibility guarantees are possible if the uncertainty bound is known initially. RMPC has successfully been applied to safety-critical applications such as automated driving \cite{SolopertoEtalMueller2019}, autonomous racing \cite{WischnewskiEtalLohmann2021}, and robotic manipulation \cite{NubertEtalTrimpe2020}.

\gls{SMPC} reduces conservatism by employing chance constraints \cite{Mesbah2016, FarinaGiulioniScattolini2016}. Chance constraints allow for a small probability of constraint violation, reducing the impact of unlikely worst-case uncertainty realizations. Multiple SMPC approaches exist to determine a tractable reformulation of the probabilistic chance constraint, e.g., analytical reformulations based on normal distributions \cite{FarinaEtalScattolini2015}, sampling based approaches \cite{BlackmoreEtalWilliams2010, SchildbachEtalMorari2014}, affine disturbance feedback approaches \cite{GoulartKerriganMaciejowski2006}, or tube-based approaches \cite{LorenzenEtalAllgoewer2017}. Applications to safety-critical systems mainly focus on automated driving \cite{CarvalhoEtalBorrelli2014, CesariEtalBorrelli2017, BruedigamEtalLeibold2020b, NairEtalBorrelli2021}. However, whereas applying these SMPC approaches yields efficient trajectories, safety is not guaranteed as the chance constraint allows for a non-zero collision probability.

Safety within SMPC is specifically addressed in \cite{BruedigamEtalLeibold2021b} for automated vehicles. This failsafe SMPC approach uses a failsafe backup predictive controller, which guarantees that the next SMPC input may be safely applied, ensuring safe SMPC trajectories for automated driving. Further approaches have recently been proposed to address safety within MPC. A combination of MPC and control barrier functions allows considering safety similarly to how Lyapunov functions are used for stability \cite{ZengZhangSreenath2021, GrandiaEtalHutter2021}. However, guaranteeing recursive feasibility in the presence of uncertainty remains a challenge. An \gls{MPC} approach to minimize constraint violation probability is proposed in \cite{BruedigamEtalLeibold2021c}, but the method is only applicable if norm-based constraints can be employed. In \cite{WabersichZeilinger2021, WabersichEtalZeilinger2021} a predictive safety filter is proposed to guarantee safety in probability for reinforcement learning. This is achieved by enforcing that only those reinforcement learning-based inputs may be applied, which allow for satisfaction of a soft-constrained \gls{ocp}.

In this work, we propose an SMPC safety algorithm for linear systems with additive uncertainty and polytopic constraints. This general safety algorithm significantly extends the approach in \cite{BruedigamEtalLeibold2021b}, which only considered one specific SMPC approach designed for automated vehicles. The safety algorithm of this work guarantees safety (satisfying all constraints) by employing a backup controller, which ensures that applying the first optimized SMPC input allows still finding a safe backup trajectory in the following step. The key contributions are as follows.
\begin{itemize}
\item We provide a safety algorithm compatible with any SMPC for linear systems with additive uncertainty and polytopic constraints. Furthermore, the risk parameter of the SMPC does not influence safety, and no terminal constraint is required in the SMPC \gls{ocp}.
\item We guarantee recursive feasibility of the safety algorithm and, in contrast to \cite{BruedigamEtalLeibold2021b}, we ensure input-to-state stability of the system origin.
\end{itemize}

With the proposed safety algorithm of this work, for a given safety-critical application and based on desired control objectives, the most suitable SMPC approach can be chosen. This choice may be made independently of required properties, which are later ensured by the proposed safety algorithm. The proposed method combines advantages of stochastic and robust predictive control. These advantages of the proposed safe SMPC algorithm are demonstrated in a simulation example, including comparisons to pure SMPC and pure RMPC.

This work is structured as follows. Section~\ref{sec:problem} introduces the problem. The safe SMPC framework and its properties are presented in Sections~\ref{sec:method} and \ref{sec:properties}. A simulation example and conclusive remarks are given in Sections~\ref{sec:results} and \ref{sec:conclusion}.

\textit{Notation:} Regular letters indicate scalars, bold lowercase letters denote vectors, and bold uppercase letters are used for matrices, e.g., $a$, $\bm{a}$, $\bm{A}$, respectively. 
The closed interval between integers $a$ and $b$ is denoted by $\mathbb{I}_{a,b}$. Absolute values and norms are indicated by $|a|$ and $||\bm{a}||$, respectively, where we consider the weighted norm $||\bm{a}||^2_A = \bm{a}^\top \bm{A} \bm{a}$. A function $\gamma$ is of class $\mathcal{K}$ if $\gamma$ is positive definite and strictly increasing. A function $\alpha$ is of class $\mathcal{K}_\infty$ if $\alpha$ is of class $\mathcal{K}$ and unbounded. Within an \gls{ocp}, the state $\bm{x}_{t+k|t}$ denotes the prediction for step $t+k$ obtained at time step $t$. We define the set addition $\mathcal{A} \oplus \mathcal{B} := \setdef[\bm{a}+\bm{b}]{\bm{a} \in \mathcal{A}, \bm{b} \in \mathcal{B}}$ and set subtraction $\mathcal{A} \ominus \mathcal{B} := \setdef[\bm{x} \in \mathbb{R}^n]{\{ \bm{x}\} \oplus \mathcal{B} \subseteq \mathcal{A}}$.
 
\section{Problem Setup}
\label{sec:problem}

We consider a linear, discrete-time system
\begin{IEEEeqnarray}{c}
\gls{xt1}  = \gls{sysA}\gls{xt}  + \gls{sysB}\gls{ut} + \gls{sysG}\gls{wt} = \bm{f}(\gls{xt},\gls{ut},\gls{wt}) \label{eq:sys}
\end{IEEEeqnarray}
with states $\gls{xt} \in \glsd{x}$, inputs $\gls{ut} \in \glsd{u}$, and uncertainties $\gls{wt} \in \glsd{w}$ at time step $t$, as well as the known matrices \gls{sysA}, \gls{sysB}, and \gls{sysG} with appropriate dimensions. System~\eqref{eq:sys} is subject to input constraints $\gls{ut} \in \gls{Ulim}$ and state constraints $\gls{xt} \in \gls{Xlim}$.

\begin{assumption}[Uncertainty]
\label{ass:uncertainty}
The uncertainty \gls{wt} is independent and identically distributed and bounded by $\gls{wt} \in \gls{Wlim}$.
\end{assumption}

The general task is to drive the state of system~\eqref{eq:sys} to the origin while keeping inputs low. In MPC, this is achieved by repeatedly solving an \gls{ocp} $\mathbb{P}(\gls{xt})$, i.e.,
\begin{IEEEeqnarray}{rll}
\IEEEyesnumber \label{eq:ocp_mpc}
\min_{\gls{Ut}}~& J(\gls{xt},\gls{Ut}) \IEEEyessubnumber\\
\textnormal{s.t. }& \bm{x}_{t+k+1|t}  = \bm{f}(\gls{xtk},\gls{utk},\gls{wtk})
\IEEEyessubnumber \IEEEeqnarraynumspace\\
& \gls{utk} \in \gls{Ulim},~~&\hspace{-35mm} k \in \gls{id}_{0,N-1} \IEEEyessubnumber\\
&\gls{xtk} \in \gls{Xlim},~~&\hspace{-35mm} k \in \gls{id}_{1,N} \IEEEyessubnumber \label{eq:xx_ocp}
\end{IEEEeqnarray}
with the finite input sequence $\gls{Ut} = (\bm{u}_{t|t}, ..., \bm{u}_{t+N-1|t})$ and the objective function
\begin{IEEEeqnarray}{c}
J(\gls{xt},\gls{Ut}) = \sum_{k=0}^{N-1} l\lr{\gls{xtk}, \gls{utk}} + \gls{Vf}(\gls{x}_{t+N|t})  \label{eq:cost_gen}
\end{IEEEeqnarray}
with prediction horizon $N$, stage cost $l$, and the terminal cost function \gls{Vf}. At each time step $t$, the first element $\gls{ut}=\bm{u}_{t|t}$ is applied to the system. This may be expressed as a control law $\gls{ut} = \gls{fl_mpc}$. 

Considering uncertainty in the state constraint \eqref{eq:xx_ocp} may lead to overly conservative results. This conservatism is reduced by chance constraints of the form
\begin{IEEEeqnarray}{c}
\mathrm{Pr}\lr{\gls{xtk} \in \gls{Xlim}} \geq \gls{rp}   \label{eq:cc}
\end{IEEEeqnarray}
with risk parameter \gls{rp}. Replacing robust constraints in RMPC by chance constraints yields an \gls{SMPC} \gls{ocp}. However, while conservatism is reduced, a small probability of constraint violation is allowed, i.e., the lower the risk parameter \gls{rp}, the higher the probability of constraint violations. 

\subsection{Property Definitions}

If \gls{MPC} is employed in safety-critical applications, three properties are required. First, safety must be ensured. Second, if the \gls{MPC} \gls{ocp} is feasible at a time step, a solution must also exist at the next time step, known as recursive feasibility of the \gls{ocp}. Third, the closed-loop system behavior must be stable. Definitions to ensure these properties are given in the following.

\begin{definition}[Safety]\label{def:safety}
The state $\bm{x}_{t_0}$ of system~\eqref{eq:sys} is safe at time step $t_0$ if it is guaranteed that there exist inputs $\bm{u}_t,~t\geq t_0$ such that the constraints $\gls{ut} \in \gls{Ulim}$ and $\gls{xt} \in \gls{Xlim}$ are satisfied for all $t \geq t_0$. 
\end{definition}

\begin{definition}[Safe Input Sequence]\label{def:safeinputs}
An input sequence $\gls{Ut} = (\gls{ut}, ..., \bm{u}_{t+N-1})$, $\bm{u}_{t+i} \in \gls{Ulim}~\forall i=\gls{id}_{0,N-1}$ is safe for system~\eqref{eq:sys} if consecutively applying the individual input elements yields the safe state sequence $(\gls{xt1}, ..., \bm{x}_{t+N})$ with individual safe states $\bm{x}_{t+i}~\forall i=\gls{id}_{1,N}$ according to Definition~\ref{def:safety}.
\end{definition}

\begin{definition}[Recursive Feasibility]\label{def:recfeas}
The \gls{MPC} \gls{ocp} \eqref{eq:ocp_mpc} for system~\eqref{eq:sys} is recursively feasible if the existence of an admissible solution $\gls{Ut}$ implies the existence of an admissible solution $\bm{U}_{t+1}$ for all $t \in \mathbb{N}$. 
\end{definition}

\begin{definition}[Robustly Positively Invariant Set]
A set \gls{roa} is robustly positively invariant for a system $\bm{f}(\gls{x},\gls{w})$ if $\bm{f}(\gls{x},\gls{w}) \in \gls{roa}$ for all $\gls{x} \in \gls{roa}$ and all $\gls{w} \in \gls{Wlim}$.
\end{definition}

\begin{definition}[Input-to-State Stability \cite{GoulartKerriganMaciejowski2006}]\label{def:ISS}
The origin of a system $\bm{f}(\gls{x}, \gls{w})$ is \gls{ISS} with region of attraction $\gls{roa} \subseteq \glsd{x}$ that contains the origin if \gls{roa} is robust positively invariant and if there exist a continuous function $V:\gls{roa}\rightarrow\mathbb{R}_{\geq 0}$ and functions $\alpha_1, \alpha_2, \alpha_3 \in \mathcal{K}_{\infty}$, $\gamma \in \mathcal{K}$ such that for all $\gls{x} \in \gls{roa}$ and $\gls{w} \in \gls{Wlim}$
\begin{IEEEeqnarray}{c}
\IEEEyesnumber
\alpha_1(||\gls{x}||) \leq V(\gls{x}) \leq  \alpha_2(||\gls{x}||) \label{eq:ISS_req1} \IEEEyessubnumber \\
V(\bm{f}(\gls{x}, \gls{w})) - V(\gls{x}) \leq - \alpha_3(||\gls{x}||) + \gamma(||\gls{w}||).  \label{eq:ISS_req2} \IEEEyessubnumber \IEEEeqnarraynumspace
\end{IEEEeqnarray}
\end{definition}
Then, function $V$ is called an \gls{ISS} Lyapunov function. If the origin of a system is \gls{ISS}, it is guaranteed that the change in $V$ is bounded as long as the uncertainty is bounded. If the uncertainty is zero, the origin of an \gls{ISS} system is asymptotically stable with region of attraction \gls{roa}.

\subsection{Problem Statement}

The aim of this work is to design a general \gls{SMPC} algorithm that exploits the advantage of reduced conservatism in \gls{SMPC} while ensuring the previously described properties.

\begin{objective}\label{prob:prob1}
The \gls{SMPC} algorithm \eqref{eq:ocp_mpc} for system~\eqref{eq:sys}, where the chance constraint \eqref{eq:cc} replaces the hard constraint~\eqref{eq:xx_ocp}, must maximize the control objective \eqref{eq:cost_gen} while guaranteeing safety (Definition~\ref{def:safety}), recursive feasibility (Definition~\ref{def:recfeas}), and stability (Definition~\ref{def:ISS}).
\end{objective}

In the following, we propose an \gls{SMPC} algorithm including a safe backup (predictive) controller that ensures satisfaction of all requirements listed in Objective~\ref{prob:prob1}.
\section{Safe Stochastic MPC}
\label{sec:method}

\subsection{General Safe SMPC Algorithm}\label{sec:framework}

\gls{SMPC} allows for a certain probability of constraint violation. Therefore, in order to use \gls{SMPC} in a safe way, it needs to be ensured that applying an \gls{SMPC} input is safe. 

We propose a general safe SMPC algorithm that consists of an \gls{SMPC} part and a backup predictive controller. This safe SMPC algorithm, shown in Figure~\ref{fig:method}, yields an input $\bm{u}_t$ at each time step $t$, which is determined based on the following two modes:
\begin{itemize}
    \item Stochastic mode (with \gls{ocp} \gls{ocp_smpc})
    \item Backup mode (with \gls{ocp} \gls{ocp_backup})
\end{itemize}
We now present details on the two \glspl{ocp} and on which mode to apply.
\begin{figure}
\vspace{1mm}
\centering
\includegraphics[width = 0.99\columnwidth]{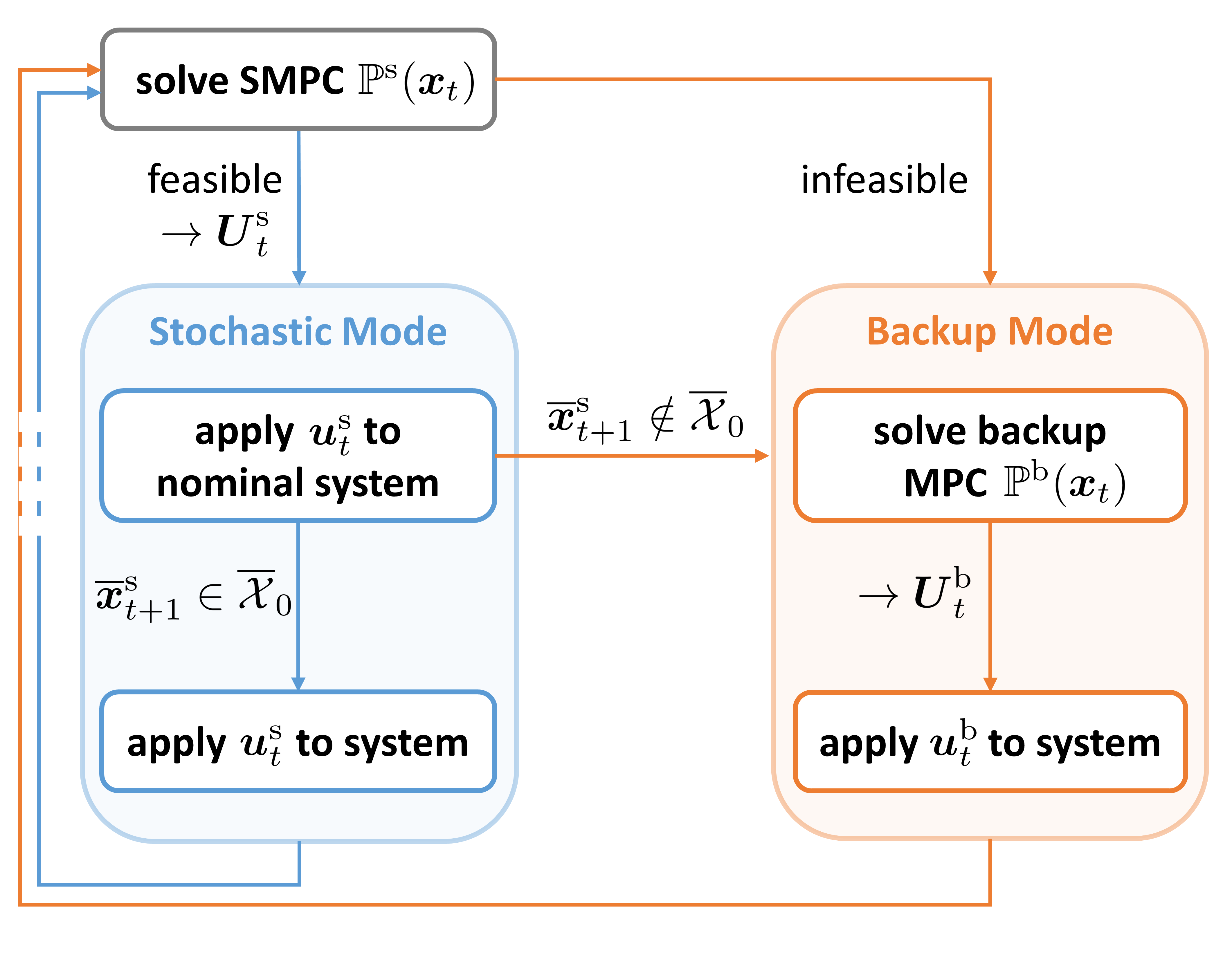}
\caption{Safe \gls{SMPC} algorithm.} 
\label{fig:method}
\end{figure}

\subsubsection{SMPC Optimal Control Problem}

We consider the general \gls{SMPC} \gls{ocp} \gls{ocp_smpc} with horizon $N$ given by
\begin{IEEEeqnarray}{rll}
\IEEEyesnumber \label{eq:ocp_smpc}
\min_{\gls{Ust}}~& J(\gls{xt},\gls{Ust}) \IEEEyessubnumber\\
\textnormal{s.t. }& \bm{x}_{t+k+1}  = \bm{f}(\gls{xtk},\gls{ustk},\gls{wtk}) \IEEEyessubnumber \IEEEeqnarraynumspace\\
& \gls{ustk} \in \gls{Ulim},~~&\hspace{-18mm}k \in \gls{id}_{0,N-1} \IEEEyessubnumber\\
&\mathrm{Pr}\lr{\gls{xtk} \in \gls{Xlim}} \geq \gls{rp},~~&\hspace{-18mm} k \in \gls{id}_{1,N} \IEEEyessubnumber \label{eq:cc_ocp}
\end{IEEEeqnarray}
yielding the optimal input sequence $\gls{Ust}{}^* = (\bm{u}^{\text{s}*}_{t|t}, ..., \bm{u}^{\text{s}*}_{t+N-1|t})$ with the \gls{SMPC} control law $\gls{fl_smpc}= \bm{u}^{\text{s}*}_{t|t}$. Any \gls{SMPC} method may be used to reformulate the chance constraint \eqref{eq:cc_ocp} into a tractable formulation, depending on the uncertainty distribution. 

\subsubsection{Backup MPC Optimal Control Problem}

We consider a backup \gls{MPC} controller with horizon $\gls{Nb}$ and \gls{ocp} \gls{ocp_backup} with cost function $J^{\text{b}}(\gls{xt},\gls{Ubt})$, yielding the optimal cost $J^{\text{b}}{}^*=J^{\text{b}}(\gls{xt},\gls{Ubt}{}^*)$ with input sequence $\gls{Ubt}{}^* = (\bm{u}^{\text{b}*}_{t|t}, ..., \bm{u}^{\text{b}*}_{t+\gls{Nb}-1|t})$, and control law $\gls{fl_backup}= \bm{u}_{t|t}^{\text{b}*}$, resulting in the closed-loop system
\begin{IEEEeqnarray}{c}
\gls{xt1}  = \bm{f}\lr{\gls{xt},\gls{fl_backup},\gls{wt}} \label{eq:sys_cl}
\end{IEEEeqnarray}
for system \eqref{eq:sys}. Various backup controllers are possible in this algorithm, given that they fulfill the following assumption. 

\begin{assumption}[Backup MPC]\label{ass:backup_reqs}
The backup MPC \gls{ocp} \gls{ocp_backup}, with value function $J^{\text{b}}{}^*$ and control law $\gls{fl_backup}$, is chosen such that \gls{ocp_backup} is recursively feasible, $\gls{xt} \in \gls{Xlim}$ and $\gls{ut} \in \gls{Ulim}$ for all $t$, and such that the origin of the closed-loop system \eqref{eq:sys_cl} is ISS with region of attraction \gls{roa}, where \gls{roa} is robust positively invariant for all $\gls{wt} \in \gls{Wlim}$.
\end{assumption}

Various MPC schemes exist that fulfill Assumption~\ref{ass:backup_reqs}, as discussed in Section~\ref{sec:MPCdetails}. 

The safe SMPC algorithm only applies SMPC inputs if it is guaranteed that the backup \gls{ocp} \gls{ocp_backup} may still be solved at the next time step. Applying the first SMPC input \gls{fl_smpc} to the nominal system yields the next nominal state
\begin{equation}
    \overline{\bm{x}}^{\text{s}}_{t+1} =  \gls{sysA}\gls{xt}  + \gls{sysB}\gls{fl_smpc}.
\end{equation}
It is guaranteed that the backup \gls{ocp} is feasible for any first step uncertainty $\gls{wt} \in \gls{Wlim}$ if \begin{equation}
    \overline{\bm{x}}^{\text{s}}_{t+1} \in \gls{roa} \ominus \gls{Wlim} \label{eq:X0minusW}
\end{equation}
where $\ol{\mathcal{X}}_{0} = \gls{roa} \ominus \gls{Wlim}$, which ensures that $\bm{x}^{\text{s}}_{t+1} \in \gls{roa} $.

\subsubsection{Safe SMPC Modes}

We are now able to propose two different modes within the safe \gls{SMPC} algorithm, evaluated at time step $t$.

\paragraph*{Stochastic Mode}
The control law $\bm{u}_t = \gls{fl_smpc}$ is applied if the \gls{SMPC} \gls{ocp} is feasible and if \eqref{eq:X0minusW} is fulfilled, i.e., $\gls{Ust}{}^* \neq \emptyset$ and $ \overline{\bm{x}}^{\text{s}}_{t+1} \in \ol{\mathcal{X}}_{0}$.

\paragraph*{Backup Mode}
If the \gls{SMPC} \gls{ocp} is infeasible or if \eqref{eq:X0minusW} is not satisfied, i.e., $\gls{Ust}{}^* = \emptyset$ or $ \overline{\bm{x}}^{\text{s}}_{t+1} \notin \ol{\mathcal{X}}_{0}$, the backup \gls{MPC} \gls{ocp} is solved and the control law $\bm{u}_t = \gls{fl_backup}$ is applied.

Note that in the stochastic mode, only one \gls{ocp} is solved, while the backup mode requires solving two \glspl{ocp}.

\subsection{MPC Details}\label{sec:MPCdetails}

The proposed algorithm allows considering any \gls{SMPC} approach to solve \eqref{eq:ocp_smpc}, e.g., SMPC with exact chance constraint reformulations based on normal distributions, affine disturbance feedback SMPC, or sampling-based SMPC. A suitable \gls{SMPC} method may be chosen depending on the type of uncertainty and the type of system.

Assumption~\ref{ass:backup_reqs} allows employing various MPC schemes for the backup controller, which enables application of a wider class of backup controllers compared to \cite{BruedigamEtalLeibold2021b}. The most intuitive choice are \gls{RMPC} approaches that guarantee recursive feasibility and stability for a bounded uncertainty. The backup MPC can also be based on other approaches, such as MPC based on reachability analysis \cite{SchuermannKochdumperAlthoff2018} or the failsafe MPC idea described in \cite{BruedigamEtalLeibold2021b}. It is even possible to consider recursively feasible \gls{SMPC} approaches as backup controllers, e.g., \cite{LorenzenEtalAllgoewer2017}, if Assumption~\ref{ass:backup_reqs} may be satisfied.

\section{Properties}
\label{sec:properties}

In the following, we show that the proposed \gls{SMPC} algorithm is recursively feasible, safe, and \gls{ISS}.

\subsection{Recursive Feasibility}

Based on Definition~\ref{def:recfeas}, we first prove recursive feasibility of the \gls{ocp} of the safe SMPC algorithm described in Section~\ref{sec:framework}.

\begin{theorem}[Recursive Feasibility]\label{th:recfeas}
Let Assumptions~\ref{ass:uncertainty} and \ref{ass:backup_reqs} hold and let the system input $\bm{u}_t$ be determined based on the proposed safe \gls{SMPC} algorithm in Section~\ref{sec:framework}. Then, for an admissible $\bm{u}_0$, obtaining a solution $\bm{u}_t$ is feasible for all $t > 0$. 
\end{theorem}

\begin{proof}
The proof is based on showing that at any time step $t$ it is guaranteed that an admissible input $\bm{u}_t$ is applied at $t$ and an admissible input $\bm{u}_{t^+}$ may be obtained at the next time step $t^+ = t+1$.

If $\gls{Ust}{}^* \neq \emptyset$ and $ \overline{\bm{x}}^{\text{s}}_{t+1} \in \gls{roa} \ominus \gls{Wlim}$, $\bm{u}_t = \gls{fl_smpc}$ exists and $\bm{x}^{\text{s}}_{t+1} \in \gls{roa}$, which guarantees that a solution $\bm{U}^{\text{b}*}_{t+1} $ exists for $\gls{ocp_backup1}$.

In the backup mode, $\bm{U}^{\text{b}*}_{t} $ exists for $\gls{ocp_backup}$ as $\gls{xt} \in \gls{roa}$ and Assumption~\ref{ass:backup_reqs} ensures that $\gls{ocp_backup1}$ remains feasible.

Hence, all possibilities are covered. This holds for all $t \in \mathbb{N}$, i.e., admissible inputs are guaranteed at subsequent time steps, which concludes the proof.
\end{proof}

Note that no terminal constraint is necessary for the SMPC \gls{ocp} to ensure recursive feasibility of the overall algorithm. Based on guaranteed recursive feasibility of the safe \gls{SMPC} algorithm, safety and stability are now discussed.

\subsection{Safety}

We require that the safe \gls{SMPC} algorithm described in Section~\ref{sec:framework} is safe. Based on Definition~\ref{def:safety}, this requirement demands that all constraints are met at all time steps, which we show in the following.

\begin{theorem}[Safety]\label{th:safety}
Let Assumptions~\ref{ass:uncertainty} and \ref{ass:backup_reqs} hold and let the system input $\bm{u}_t$ be determined based on the proposed safe \gls{SMPC} algorithm in Section~\ref{sec:framework}. Then, for a safe initial state $\bm{x}_0$, safety according to Definition~\ref{def:safety} is guaranteed for $t>0$.
\end{theorem}

\begin{proof}
The proof is based on Theorem~\ref{th:recfeas}; hence, it is guaranteed that one of the two modes is applicable at each time step $t$. In the stochastic mode, $\gls{ut} \in \gls{Ulim}$ and \eqref{eq:X0minusW} ensures that $\gls{xt} \in \gls{Xlim}$ for all $\gls{wt} \in \gls{Wlim}$. The backup mode guarantees, by design, that $\gls{ut} \in \gls{Ulim}$ and $\gls{xt} \in \gls{Xlim}$ for all $\gls{wt} \in \gls{Wlim}$. Hence, in both modes it is guaranteed that $\gls{xt} \in \gls{Xlim}$ and $\gls{ut} \in \gls{Ulim}$, which holds for all $t \in \mathbb{N}$ as the proposed \gls{SMPC} algorithm is recursively feasible.
\end{proof}

As shown, safety is ensured by the backup predictive controller and \eqref{eq:X0minusW}, despite SMPC allowing for constraint violations in the open-loop prediction.

\subsection{Stability}

In MPC, stability is often proved by showing that the value function is decreasing for subsequent time steps, also known as the descent property. These proofs are based on the MPC idea of a moving horizon, where the previously planned input sequence remains valid and only one additional input element is added to the input sequence for the next time step. For the proposed safe \gls{SMPC} algorithm, however, this assumption does not hold. Since switching between different modes is possible, the predicted input and state trajectories may vary at each time step. We tackle this challenge by using an \gls{ISS} result from \cite[Lemma~22]{GoulartKerriganMaciejowski2006}, which we repeat in the following lemma.

\begin{lemma}[Lipschitz ISS Lyapunov Function \cite{GoulartKerriganMaciejowski2006}]\label{lem:ISS_Lipschitz}
Let $\bm{f}: \gls{roa} \times \gls{Wlim} \rightarrow \glsd{x}$ be Lipschitz continuous on $\gls{roa} \times \gls{Wlim}$. Let $\gls{roa} \in \glsd{x}$ contain the origin and be a robust positively invariant set for the function $\bm{f}(\gls{x},\gls{w})$. Let there exist a Lipschitz continuous function $V: \gls{roa} \rightarrow \mathbb{R}_{\geq 0}$ such that for all $\gls{x} \in \gls{roa}$
\begin{IEEEeqnarray}{c}
\IEEEyesnumber
\alpha_1(||\gls{x}||) \leq V(\gls{x}) \leq  \alpha_2(||\gls{x}||)  \IEEEyessubnumber \label{eq:ISS_req1b} \\
V(\bm{f}(\gls{x}, \bm{0})) - V(\gls{x}) \leq - \alpha_3(||\gls{x}||)   \IEEEyessubnumber \IEEEeqnarraynumspace \label{eq:ISS_req2b}
\end{IEEEeqnarray}
with functions $\alpha_1, \alpha_2, \alpha_3 \in \mathcal{K}_{\infty}$. Then, $V$ is an \gls{ISS} Lyapunov function and the origin is \gls{ISS} for system $\bm{f}(\gls{x},\gls{w})$ with region of attraction \gls{roa}.
\end{lemma}

Lemma~\ref{lem:ISS_Lipschitz} ensures that the origin of a system subject to uncertainty is \gls{ISS} if the undisturbed system is asymptotically stable and the system is Lipschitz continuous with respect to state \gls{x} and uncertainty \gls{w}.

\begin{assumption}[Backup MPC Cost]\label{ass:cost_backup}
The cost function $J^{\text{b}}$ is selected according to \eqref{eq:cost_gen}. The stage cost is chosen as $l\lr{\gls{xtk}, \gls{utk}} =  ||\ol{\bm{x}}_{t+k|t}||^2_{\bm{Q}} +  ||\gls{utk}||^2_{\bm{R}}$ 
with $\bm{Q} = \bm{Q}^\top \succ 0$, $\bm{R} = \bm{R}^\top \succ 0$, and the nominal states $\ol{\bm{x}}_{t+k|t}$. The terminal cost $\gls{Vf}(\gls{x}_{t+N|t})$ is chosen as a Lyapunov function in a terminal set \gls{Xf} for the undisturbed closed-loop system $\gls{xt1} = (\gls{sysA} + \gls{sysB} \bm{K}) \gls{xt}$ such that for all $\gls{xt} \in \gls{Xf}$
\begin{equation}
    \gls{Vf}((\gls{sysA} + \gls{sysB} \bm{K})\gls{xt}) - \gls{Vf}(\gls{xt}) \leq - \gls{xt}^\top(\bm{Q} + \bm{K}^\top \bm{R} \bm{K} ) \gls{xt}
\end{equation}
where $\bm{K}$ is a stabilizing feedback matrix.
\end{assumption}

We can now formulate the \gls{ISS} property of a system controlled with the proposed algorithm.

\begin{theorem}[\textit{ISS for Safe SMPC}]\label{th:ISS}
Let Assumptions~\ref{ass:uncertainty}-\ref{ass:cost_backup} hold and let the system input $\bm{u}_t$ be determined based on the proposed safe \gls{SMPC} algorithm in Section~\ref{sec:framework}. Then, for $\bm{x}_0 \in \gls{roa}$, the origin is \gls{ISS} for system \eqref{eq:sys} and $\bm{x}_t \in \gls{roa},~~t>0$.
\end{theorem}

\begin{proof}
We prove \gls{ISS} by showing that $V(\gls{xt}) = J^{\text{b}}{}^*$ is an \gls{ISS} Lyapunov function for \eqref{eq:sys} where $V(.)$ satisfies \eqref{eq:ISS_req1b} and \eqref{eq:ISS_req2b}. Any input prediction in either of the two modes can be described by $\gls{utk}' = \gls{utk}^{\text{b}} - \tilde{\bm{u}}_{t+k|t},~k \in \gls{id}_{0,N-1}$ where $\tilde{\bm{u}}_{t+k|t}$ represents the offset between the backup MPC input element $\gls{utk}^{\text{b}}$ and the input element obtained in the safe \gls{SMPC} algorithm $\gls{utk}'$. As $\gls{utk}'$ and $\gls{utk}^{\text{b}}$ are bounded, $\tilde{\bm{u}}_{t+k|t}$ is bounded, allowing to define the new bounded uncertainty $\gls{w2tk} = (\tilde{\bm{u}}_{t+k|t}^\top, \gls{wtk}^\top)^\top$. This yields the closed loop system
\begin{IEEEeqnarray}{rl}
\IEEEyesnumber
\bm{f}(\gls{xt},\gls{ut},\gls{wt}) &= \gls{sysA}\gls{xt}  + \gls{sysB}\gls{ut} + \gls{sysG}\gls{wt} \IEEEyessubnumber \\
&= \gls{sysA}\gls{xt}  + \gls{sysB}\gls{fl_backup} -\gls{sysB}\tilde{\bm{u}}_{t} + \gls{sysG}\gls{wt} \IEEEyessubnumber \IEEEeqnarraynumspace \\
&= \gls{sysA}\gls{xt}  + \gls{sysB}\gls{fl_backup} +[-\gls{sysB}, \gls{sysG}] \gls{w2t} \IEEEyessubnumber \IEEEeqnarraynumspace 
\end{IEEEeqnarray}
which can be abbreviated by $\bm{f}'(\gls{xt},\gls{w2t})$.

Function $\bm{f}'$ is continuous and $\bm{f}'(\bm{0},\bm{0}) = 0$. With Assumption~\ref{ass:cost_backup}, it holds that $V(\cdot)$ is positive definite and continuous on \gls{roa}. Hence, based on \cite[Lemma 4.3]{Khalil2002}, functions $\alpha_1, \alpha_2 \in \mathcal{K}_{\infty}$ exist such that $\alpha_1(||\gls{xt}||) \leq V(\gls{xt}) \leq  \alpha_2(||\gls{xt}||)$, i.e., \eqref{eq:ISS_req1b} is fulfilled.

Due to Assumption~\ref{ass:backup_reqs},  $V(\gls{xt}) = J^{\text{b}}{}^*$ is an \gls{ISS} Lyapunov function for the undisturbed system with $\gls{w2t} = \bm{0}$, i.e., $V(\bm{f}'(\gls{xt}, \bm{0})) - V(\gls{xt}) \leq - \alpha_3(||\gls{xt}||)$. With $J^{\text{b}}$ designed according to Assumption~\ref{ass:cost_backup} and a bounded \gls{roa}, $J^{\text{b}}{}^*$ is Lipschitz continuous. Hence, Lemma~\ref{lem:ISS_Lipschitz} is fulfilled and $V(\cdot)$ is an \gls{ISS} Lyapunov function for $\bm{f}'(\gls{xt},\gls{w2t})$ with $\gls{xt} \in \gls{roa}$, i.e., it holds that $V(\bm{f}'(\gls{xt}, \gls{w2t})) - V(\gls{xt}) \leq - \alpha_3(||\gls{xt}||) + \gamma(||\gls{w2t}||)$. Hence, the origin is \gls{ISS} with the safe \gls{SMPC} algorithm. 
\end{proof}

Note that tuning the risk parameter in the \gls{SMPC} \gls{ocp} does not impact recursive feasibility, safety, or stability. This allows choosing a risk parameter that yields the most efficient behavior.

\section{Numerical Results}
\label{sec:results}

We analyze the proposed algorithm in a numerical example, based on \cite{LorenzenEtalAllgoewer2017} and elaborate on the advantages over \gls{SMPC} and \gls{RMPC}. Simulations are carried out in Matlab where the set calculations are done with the Mutli-Parametric Toolbox 3 \cite{HercegEtalMorari2013} and the \gls{MPC} routine is based on \cite{GruenePannek2017}.

\subsection{Simulation Setup}
\label{sec:simusetup}

We consider the discrete-time system
\begin{equation}
    \gls{xt1} = \begin{bmatrix} 1 & 0.0075 \\ -0.143 & 0.996\end{bmatrix} \gls{xt} + \begin{bmatrix} 4.798 \\ 0.115\end{bmatrix} u_t + \begin{bmatrix} 1 & 0 \\ 0 & 1\end{bmatrix} \gls{wt}
\end{equation}
with $\gls{x} = (x_1, x_2)^\top$ and the normally distributed uncertainty $\gls{w} \sim \mathcal{N} \lr{\bm{0}, \gls{Sigmaw}}$, $\gls{wt} \in \gls{Wlim}$ where $\gls{Sigmaw} = \mathrm{diag}(0.06, 0.06)$ and $\gls{Wlim} = \setdef[\gls{wt}]{||\gls{wt}||_{\infty} \leq 0.07}$. The input is bounded by $|u_t| \leq 0.2$ and we employ the state constraint $x_1 \leq 2.8$. Additionally, we define $|x_1| \leq 10$ and $|x_2|\leq 10$ to obtain a bounded set \gls{Xlim}, however, in the following simulation only $x_1 \leq 2.8$ is regarded. The initial state is $\gls{x0} = (-1.3,~3.5)^\top$.

For the \gls{SMPC} \gls{ocp}, we approximate the uncertainty with the non-truncated normal distribution $\gls{w} \sim \mathcal{N} \lr{\bm{0}, \gls{Sigmaw}}$ and split the state into a deterministic and a probabilistic part $\gls{xt} = \bm{z}_t + \bm{e}_t$, yielding an adapted input $u_t = \bm{K} \gls{xt} + \nu_t$ with a stabilizing feedback matrix $\bm{K}$ and the new input decision variable $\nu_t$. The state constraint is considered as the chance constraint $\mathrm{Pr}\lr{x_1 \in 2.8} \geq \gls{rp} $ with $\gls{rp} = 0.8$. The normal distribution \gls{w} allows for the chance constraint reformulation
\begin{IEEEeqnarray}{rl}
\IEEEyesnumber \label{eq:tightening}
x_{1,k} &\leq 2.8 - \gls{tightening} \IEEEyessubnumber \\
\gls{tightening} &= \sqrt{2 [1,0]^\top \gls{Sigmaek} [1,0]} \erf^{-1} (2 \gls{rp} -1) \IEEEyessubnumber
\end{IEEEeqnarray}
with the inverse error function $\erf^{-1}(.)$ and the error covariance matrix 
$\gls{Sigmaek1} = \bs{\Phi} \gls{Sigmaek} \bs{\Phi}^\top + \gls{Sigmaw}$ 
with $\gls{Sigmae0} = \mathrm{diag}(0,0)$ and $\bs{\Phi} = \gls{sysA} + \gls{sysB} \bm{K} $.

For the backup MPC, we use an \gls{RMPC} approach according to \cite{MayneLangson2001}, satisfying Assumption~\ref{ass:backup_reqs}. This approach yields the tightened state constraint $\ol{x}_1 \leq 1.72$ and tightened input constraint $-0.018\leq \ol{u}_t \leq 0.025$. The terminal constraint \gls{Xf} is chosen to be a maximal robust control invariant set.

For \gls{SMPC} and \gls{RMPC}, we employ a sampling time $\Delta t = 0.1$, a horizon $N = \gls{Nb} = 11$, and we use the stabilizing feedback gain $\bm{K} = [-0.29, 0.49]$.

For both the \gls{SMPC} and \gls{RMPC} \gls{ocp}, we use the cost according to Assumption~\ref{ass:cost_backup} with $\bm{Q} = \mathrm{diag}(1,10)$ and $\bm{R}~=~1$. We choose $\gls{Vf}(\bm{x}) = ||\bm{x}||^2_{\bm{Q}_{\text{f}}}$ with $\bm{Q}_{\text{f}} = \begin{bmatrix} 1.91 & -5.06 \\ -5.06 & 39.54\end{bmatrix}$, which satisfies the discrete-time algebraic Riccati equation.

\subsection{Simulation Results}

We first analyze the resulting trajectories of a simulation with zero uncertainty using the proposed safe SMPC algorithm as well as pure RMPC and SMPC, where $\gls{rp} = 0.8$. The pure RMPC and pure SMPC are based on the backup RMPC controller and the SMPC controller described in Section~\ref{sec:simusetup}, respectively. The results are shown in Figure~\ref{fig:results_nocc}.

\begin{figure}
\vspace{1.5mm}
\centering
\includegraphics[width = 0.99\columnwidth]{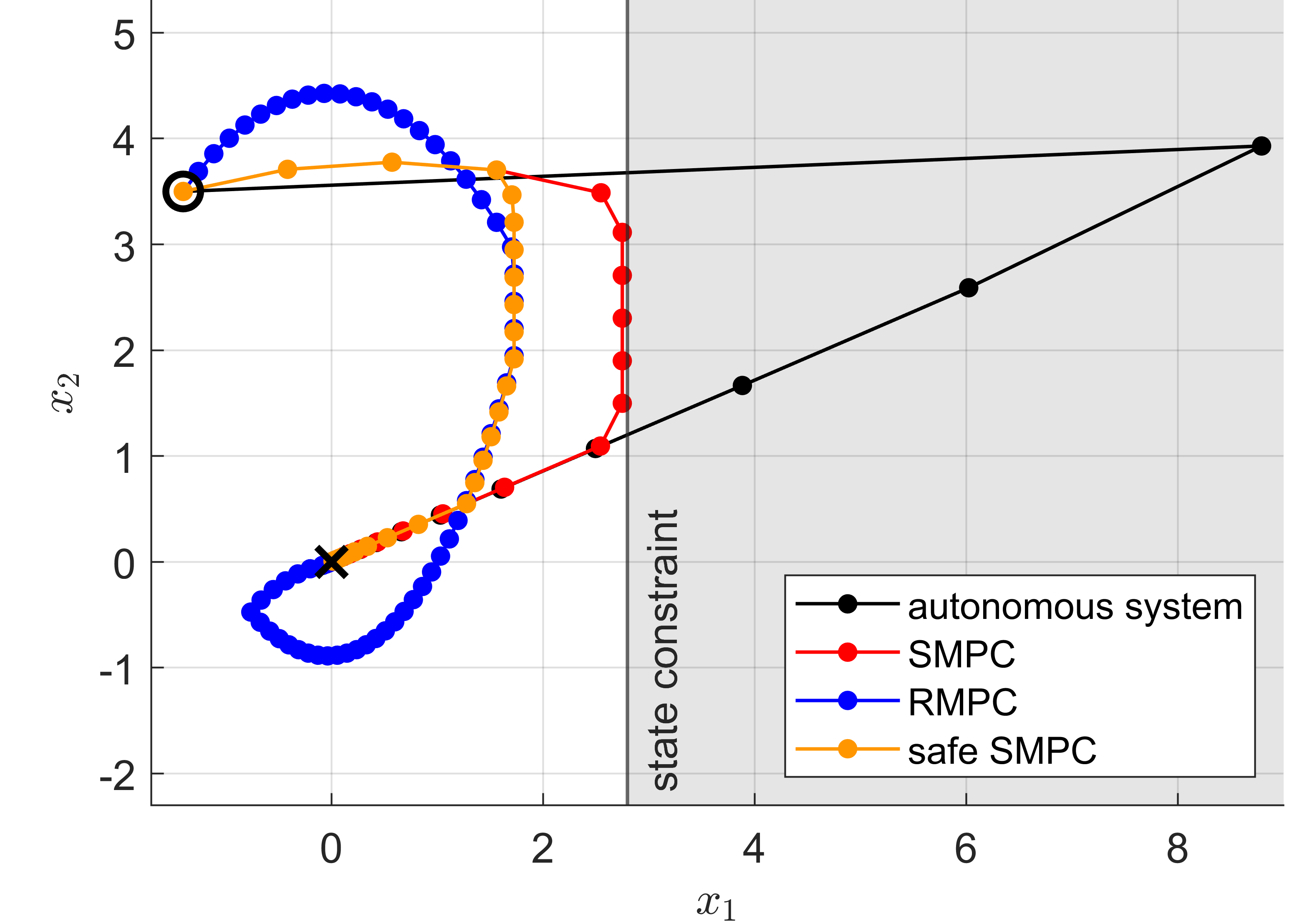}
\caption{Simulation results for the autonomous system, RMPC, SMPC, and safe SMPC, all without uncertainty.} 
\label{fig:results_nocc}
\end{figure}

Whereas the autonomous system with $u_t=0$ violates the constraint $x_1 \leq 2.8$, pure SMPC moves as close towards the state constraint as the constraint tightening \gls{tightening} allows. Considering worst-case uncertainty yields tighter state and input constraints for pure RMPC, requiring more steps to reach the origin. The safe SMPC approach is initially similar to pure SMPC, as $x_1$ is far from the state constraint. Once \eqref{eq:X0minusW} is not satisfied anymore, the safety mechanism is triggered and the backup mode becomes active. 
For the following steps, the input of the safe backup RMPC is applied until it is possible that SMPC inputs satisfy \eqref{eq:X0minusW} again.

The procedure of the safe SMPC approach is illustrated in Figure~\ref{fig:results_safeSMPC}, showing the resulting trajectories of 10 simulation runs subject to uncertainty. 
\begin{figure}
\vspace{1mm}
\centering
\includegraphics[width = 0.94\columnwidth]{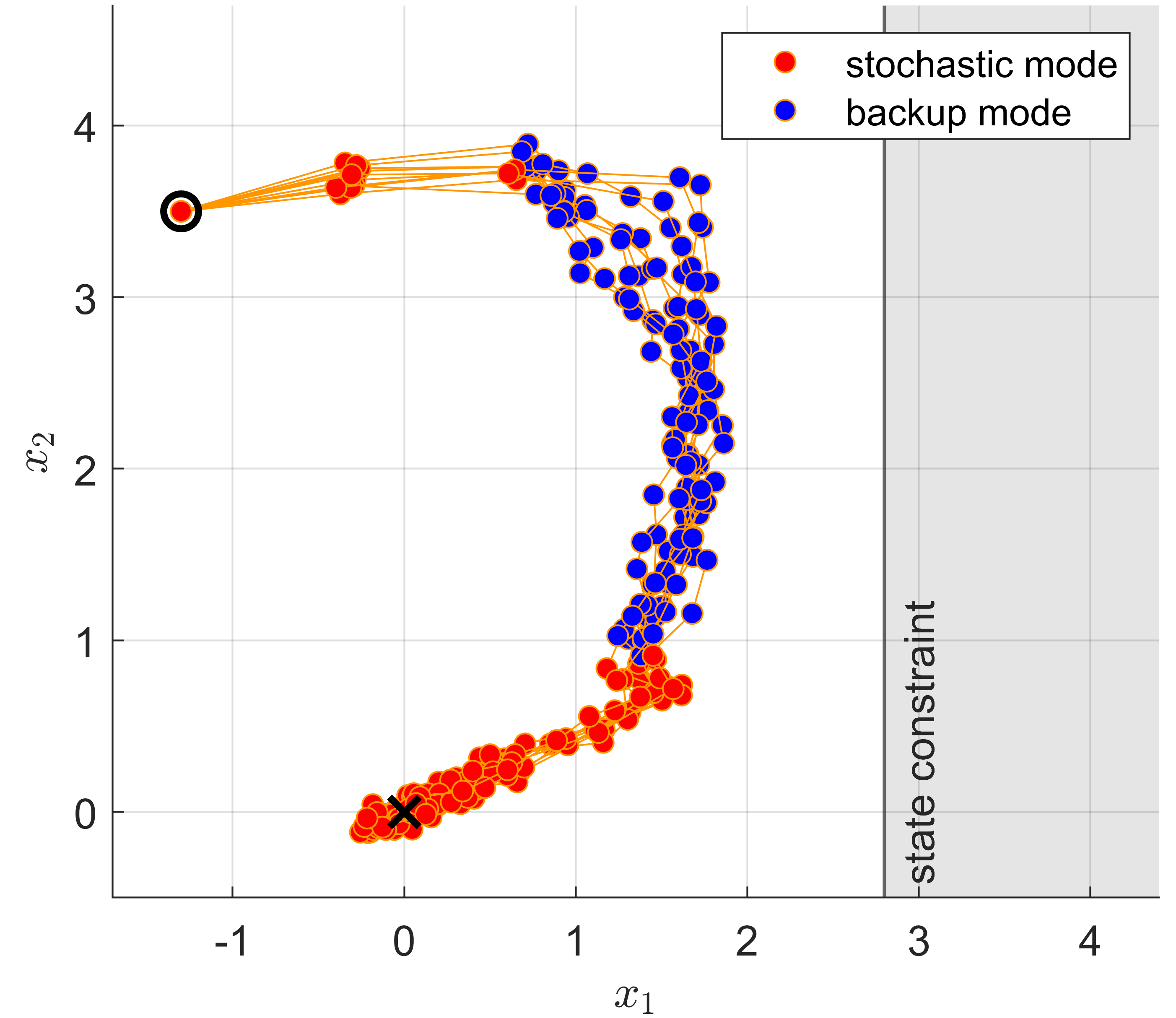}
\caption{Resulting safe SMPC trajectories for 10 simulation runs with uncertainty. 
} 
\label{fig:results_safeSMPC}
\end{figure}
Switching to the backup RMPC inputs ensures that the state constraint is never violated.

Analyzing 100 simulation runs of each safe SMPC, pure RMPC, and pure SMPC with $N_\text{sim} = 80$ simulation steps highlights the advantages of the proposed method. The results are given in Table~\ref{tab:comparison}. The cost is determined by
\begin{IEEEeqnarray}{c}
J_\text{sim} = \sum_{k=1}^{N_\text{sim}} \norm{ \bm{x}_k}^2_{\bm{Q}} + \norm{\bm{u}_{k-1}}^2_{\bm{R}} .\label{eq:overallcost}
\end{IEEEeqnarray}
\begin{table}
\centering
\begin{tabular}{l    c     c }
\toprule
method &  avg. cost & avg. violations per run \\
\midrule
RMPC & $3.56\mathrm{e}3$  & 0 \\
\addlinespace
SMPC & $0.88\mathrm{e}3$ & 0.89 \\
\addlinespace
safe SMPC & $1.13\mathrm{e}3$ & 0 \\
\bottomrule
\end{tabular}
\caption{Comparison.}
\label{tab:comparison}
\end{table}

Whereas SMPC yields the lowest cost, constraint violations occur. RMPC avoids constraint violations but the cost is significantly increased. Safe SMPC guarantees constraint satisfaction with only slightly higher cost compared to SMPC, combining the advantages of SMPC and RMPC. In summary, the proposed algorithm ensures robustness when constraints are active, but allows reducing the conservative behavior of the backup controller when constraints are inactive. 

\section{Conclusion}
\label{sec:conclusion}

The proposed safe SMPC algorithm offers the possibility to combine optimistic SMPC planning with a safety guarantee. In addition, recursive feasibility and stability is guaranteed, without the need of an SMPC terminal constraint. 

The proposed algorithm is not limited to SMPC. Instead of using SMPC, other controllers, e.g., learning-based methods, can be used. This would allow ensuring safety and stability for learning-based controllers.

\section*{Acknowledgement}

The authors thank Francesco Borrelli and Sarah Buhlmann for the collaboration and valuable discussions. This work was supported by a fellowship within the IFI program of the German Academic Exchange Service (DAAD) and the Bavaria California Technology Center (BaCaTeC) grant 1-[2020-2].

\bibliography{./references/Dissertation_bib}
\bibliographystyle{unsrt}

\end{document}